\documentclass{amsart}

\usepackage{graphicx}
\usepackage{amsthm}

\newtheorem{proposition}{Proposition}

\begin{document}

\title[Energy simulation in double sine-Gordon]{On the simulation of the energy transmission in the forbidden band-gap of a spatially discrete double sine-Gordon system}

\author{J. E. Mac\'{\i}as-D\'{\i}az}
\address{Departamento de Matem\'{a}ticas y F\'{\i}sica, Universidad Aut\'{o}noma de Aguascalientes, Avenida Universidad 940, Ciudad Universitaria, Aguascalientes 20131, Mexico}
\email{jemacias@correo.uaa.mx}


\begin{abstract}
In this work, we present a numerical method to consistently approximate solutions of a spatially discrete, double sine-Gordon chain which considers the presence of external damping. In addition to the finite-difference scheme employed to approximate the solution of the difference-differential equations of the model under investigation, our method provides positivity-preserving schemes to approximate the local and the total energy of the system, in such way that the discrete rate of change of the total energy with respect to time provides a consistent approximation of the corresponding continuous rate of change. Simulations are performed, first of all, to assess the validity of the computational technique against known qualitative solutions of coupled sine-Gordon and coupled double sine-Gordon chains. Secondly, the method is used in the investigation of the phenomenon of nonlinear transmission of energy in double sine-Gordon systems; the qualitative effects of the damping coefficient on the occurrence of the nonlinear process of supratransmission are briefly determined in this work, too.
\end{abstract}

\maketitle

\section{Introduction\label {S:Intro}}

The well-known sine-Gordon equation is a partial differential equation that appears in many applications, either in its original form or as a slight modification of the classical version. For instance, a damped sine-Gordon equation appears in the study of long Josephson junctions between superconductors when dissipative effects are taken into account \cite {Solitons}. A similar partial differential equation with different nonlinear term appears in the study of fluxons in Josephson transmission lines \cite {Lomdahl}. Meanwhile, a modified Klein-Gordon equation appears in the statistical mechanics of nonlinear coherent structures ---such as solitary waves---, in the form of a Langevin equation (see \cite {Makhankov}, pp. 298--309).

The spatially discrete version of the sine-Gordon equation also has many important applications. For instance, a coupled system of discrete sine-Gordon equations may describe a chain of harmonic oscillators coupled through Hookean springs \cite {Geniet-Leon} or a system of Josephson junctions attached through superconducting wires \cite {Chevrieux2}. In the former case, a system initially at rest, with void initial velocities and sinusoidal Dirichlet boundary condition at one end, is employed to study the phenomenon of supratransmission of energy \cite {Geniet-Leon, Geniet-Leon2}, which is a nonlinear process characterized by a sudden increase in the amplitude of wave signals generated by the perturbed boundary, for driving amplitudes above a critical value called the {\em supratransmission threshold}. This phenomenon is also present in the investigation of the transmission of energy in chains of Josephson junctions, except that, in this case, a harmonic Neumann boundary condition needs to be imposed upon the physical problem for the sake of meaningfulness \cite {Chevrieux2}.

We must remark, in this point, that several other bounded, nonlinear regimes present the phenomenon of supratransmission when they are harmonically perturbed at one end, such is the case of Klein-Gordon arrays \cite {Geniet-Leon}, Fermi-Pasta-Ulam systems \cite {Khomeriki}, Bragg media in the nonlinear Kerr regime \cite {Leon-Spire}, and even in spatially continuous, bounded media described by undamped sine-Gordon equations \cite {Khomeriki-Leon}. Moreover, the presence of the phenomenon of nonlinear supratransmission is also found in discrete, double sine-Gordon chains \cite {Geniet-Leon2}. However, the specialized mathematical literature unfortunately lacks studies to approximate the occurrence of the process in these systems.

Nevertheless, in this work, we study the process of nonlinear supratransmission in dissipative, double sine-Gordon chains, employing a dissipation-preserving finite-difference scheme. Of course, there exist many analytical \cite {ravi2009differential, wang2009discrete} and numerical \cite {dehghan2009numerical, aydin2007symplectic, cassol2008numerical} techniques to approximate solutions of Klein-Gordon-like equations. The computational method employed in this work distinguishes from many other techniques available in the literature in that it consistently approximates not only the solution of the physical model under study, but also the local energy density, the total energy, and the rate of change of the energy with respect to time. Moreover, the nonnegative character of both the energy density and the total energy ---a characteristic which is not preserved by other numerical techniques \cite {Macias-Puri}---, is preserved by our method. As we shall see later on, these qualities are, by definition, highly desirable characteristics of any computational technique employed in the study of the phenomenon of supratransmission.

In Section \ref {S:Preliminaries}, we present the system of ordinary differential equations that motivates our study, together with the local energy functions associated, and the total energy of the system. A proposition which summarizes the expression of the derivative of the total energy with respect to time is provided in this stage, as well as a brief description of the process of nonlinear supratransmission, particularly, in double sine-Gordon systems. Section \ref {S:NumMeth} introduces the numerical method employed to approximate the solutions of the model under investigation, the local energy distribution, and the total energy of the system. A subsection on the numerical properties of the method summarizes the properties of consistency established for convenience in the appendices, and another subsection presents some remarks on the computational implementation of our technique. Section \ref {S:Simul} presents simulations of sine-Gordon and double sine-Gordon systems, obtained by means of our method. The former regime is employed only for validation purposes, while the simulations on the latter (which follow the same methodology proposed in \cite {Geniet-Leon}) are aimed at establishing the existence of the process of nonlinear supratransmission in this system. Finally, this work closes with a section of concluding remarks and further directions of research.

\section{Preliminaries\label {S:Preliminaries}}

\subsection{Physical model\label {SS:PhysModel}}

Throughout this work, we let $\mathbb {Z} _N = \{ 1 , 2 , \dots , N - 1 \}$, for every  positive integer $N$; obviously, we will assume that $N > 1$ for the sake of non-triviality. Moreover, we let $\overline {\mathbb {Z}} _N = \mathbb {Z} _N \cup \{ 0 , N \}$, that is, $\overline {\mathbb {Z}} _N = \{ 0 , 1 , 2 , \dots , N \}$.

Let $N$ be a positive integer and, for every $n = 0 , 1 , \dots , N$, let $u _n$ be a real function of time $t \geq 0$. Moreover, let $c$ be a positive real number, and let $\gamma$ be a nonnegative number. Throughout, we consider a mechanical chain of nonlinear oscillators obeying the system of ordinary differential equations with initial-boundary conditions
\begin{equation}
	\begin{array}{c}
		\displaystyle {\frac {d ^2 u _n} {d t ^2} - \delta _x ^{(2)} u _n + \gamma \frac {d u _n} {d t} + V ^\prime (u _n) = 0}, \quad \forall n \in \mathbb {Z} _N, \\
		\left\{ \begin{array}{ll}
			u _n (0) = 0, & \forall n \in \mathbb {Z} _N, \\
			\displaystyle {\frac {d u _n} {d t} (0) = 0}, & \forall n \in \mathbb {Z} _N, \\
			u _0 (t) = \phi (t), & \forall t \geq 0, \\
			u _N (t) - u _{N - 1} (t) = 0, & \forall t \geq 0.
		\end{array} \right.
	\end{array}
	\label{Eq:Model}
\end{equation}
In other words, we consider a spatially discrete, bounded system initially at rest, with zero initial velocities, perturbed at the left end by a function $\phi$ which we will assume to be continuous, and with discrete Neumann boundary condition on the right end. The constant $\gamma$ is immediately identified as the external damping coefficient, while $c$ is clearly the coupling coefficient between nodes. Here, the spatial, second-difference operator 
\begin{equation}
	\delta _x ^{(2)} u _n = c ^2 (u _{n + 1} - 2 u _n + u _{n - 1}), \quad \forall n \in \mathbb {Z} _N,
\end{equation}
has been employed for convenience.

For the sake of concreteness, we will consider a driving function of the form
\begin{equation}
	\phi (t) = A \sin (\Omega t),
	\label{Eq:Driving}
\end{equation}
where the driving amplitude $A$ and the driving frequency $\Omega$ are positive real numbers. Moreover, we consider a potential function of the form
\begin{equation}
	V (u) = \frac {1} {2} - \frac {1} {6} \left[ 2 \cos u + \cos (2 u) \right],
	\label{Eq:Potential}
\end{equation}
whence the double sine-Gordon law $V ^\prime (u) = \frac {1} {3} \left[ \sin u + \sin (2 u) \right]$ readily results. Indeed, let $c = \frac {1} {\Delta x}$. If $\Delta x$ is relatively small (or, equivalently, $c$ is relatively large), then the system of ordinary differential equations of (\ref {Eq:Model}) approximates the spatially continuous, partial differential equation
\begin{equation}
	\frac {\partial ^2 v} {\partial t ^2} - \frac {\partial ^2 v} {\partial x ^2} + \gamma \frac {\partial v} {\partial t} + V ^\prime (v) = 0, \quad x \in [0 , L],
	\label{Eq:DSGPDE}
\end{equation}
where $L = N \Delta x$, and $v$ is a function of space $x$ and time $t$. This equation is clearly identified with the classical, double sine-Gordon equation with constant external damping. 

Of course, different potentials may give rise to other important models in mathematical physics. For instance, $V (u) = 1 - \cos (u)$ is the potential for the sine-Gordon regime, while $V (u) = \frac {1} {2!} u ^2 - \frac {1} {4!} u ^4 + \frac {1} {6!} u ^6$ corresponds to the potential of a nonlinear Klein-Gordon equation. In fact, it is important to warn the reader that the dissipation-preserving numerical technique presented in this work is valid not only for the double sine-Gordon potential, but also for any differentiable function $V$ defined on the real line.

\subsection{Energy of the system\label {SS:Energy}}

Let $n \in \mathbb {Z} _N$. For physical reasons, it is important to notice that the local energy of the $n$th node in the undamped system governed by (\ref {Eq:Model}) is provided by the expression
\begin{equation}
	H _n = \frac {1} {2} \left[ \left( \frac {d u _n} {d t} \right) ^2 + \left( \delta _x u _n \right) ^2 \right] + V (u _n), \quad \forall n \in \mathbb {Z} _N,
	\label {Eq:LocalEnergy}
\end{equation}
where the spatial, first-order difference operator $\delta _x$ is defined through 
\begin{equation}
	\delta _x u _n = c \left( u _{n + 1} - u _n \right), \quad \forall n \in \mathbb {Z} _N.
\end{equation}

In these terms, the total energy $E$ of the system (\ref {Eq:Model}) is obtained by adding the local energies $H _n$, for $n \in \mathbb {Z} _N$, and the potential due to the coupling in the boundaries of the chain. In other words,
\begin{equation}
	E = \sum _{n \in \mathbb {Z} _N} H _n + \frac {1} {2} \left( \delta _x u _0 \right) ^2.
	\label {Eq:TotalEnergy}
\end{equation}

The following proposition is easy to establish.

\begin{proposition}
	The rate of change of energy with respect to time of a system governed by (\ref {Eq:Model}) is given by
	\begin{equation}
		\frac {d E} {d t} = - c \left( \delta _x u _0 \right) \frac {d u _0} {d t} - \gamma \sum _{n \in \mathbb {Z} _N} \left( \frac {d u _n} {d t} \right) ^2.
		\label{Eq:DerivativeEnergy}
	\end{equation} \qed
\end{proposition}

As a corollary, the system (\ref {Eq:Model}) conserves the total energy if no external damping is present and, either $\phi$ is a constant function or a void Neumann condition is imposed on the left end of the chain.

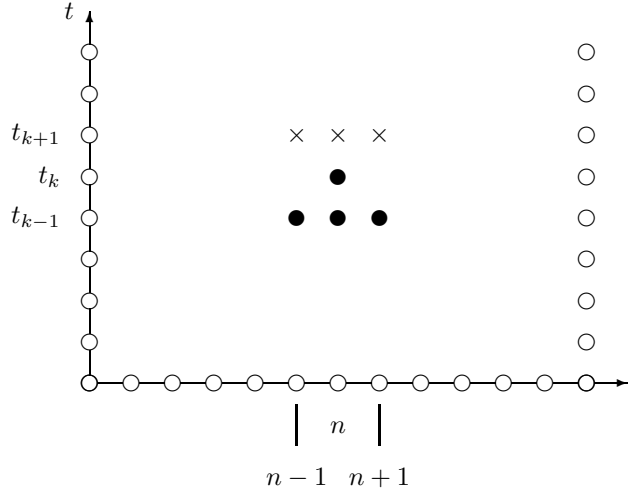
\begin{figure}
\begin{center}
\setlength{\unitlength}{.55mm}
\begin{picture}(140,100)
\put(10,12){\line(0,1){6}} %
\put(10,22){\line(0,1){6}} %
\put(10,32){\line(0,1){6}} %
\put(10,42){\line(0,1){6}} %
\put(10,52){\line(0,1){6}} %
\put(10,62){\line(0,1){6}} %
\put(10,72){\line(0,1){6}} %
\put(10,82){\line(0,1){6}} %
\put(10,92){\vector(0,1){8}} %
\put(5,100){\makebox(0,0)[5]{$t$}} %
\put(12,10){\line(1,0){6}} %
\put(22,10){\line(1,0){6}} %
\put(32,10){\line(1,0){6}} %
\put(42,10){\line(1,0){6}} %
\put(52,10){\line(1,0){6}} %
\put(62,10){\line(1,0){6}} %
\put(72,10){\line(1,0){6}} %
\put(82,10){\line(1,0){6}} %
\put(92,10){\line(1,0){6}} %
\put(102,10){\line(1,0){6}} %
\put(112,10){\line(1,0){6}} %
\put(122,10){\line(1,0){6}} %
\put(132,10){\vector(1,0){8}} %
\put(10,10){\circle{4}} %
\put(10,20){\circle{4}} %
\put(10,30){\circle{4}} %
\put(10,40){\circle{4}} %
\put(10,50){\circle{4}} %
\put(10,60){\circle{4}} %
\put(10,70){\circle{4}} %
\put(10,80){\circle{4}} %
\put(10,90){\circle{4}} %
\put(130,10){\circle{4}} %
\put(130,20){\circle{4}} %
\put(130,30){\circle{4}} %
\put(130,40){\circle{4}} %
\put(130,50){\circle{4}} %
\put(130,60){\circle{4}} %
\put(130,70){\circle{4}} %
\put(130,80){\circle{4}} %
\put(130,90){\circle{4}} %
\put(10,10){\circle{4}} %
\put(20,10){\circle{4}} %
\put(30,10){\circle{4}} %
\put(40,10){\circle{4}} %
\put(50,10){\circle{4}} %
\put(60,10){\circle{4}} %
\put(70,10){\circle{4}} %
\put(80,10){\circle{4}} %
\put(90,10){\circle{4}} %
\put(100,10){\circle{4}} %
\put(110,10){\circle{4}} %
\put(120,10){\circle{4}} %
\put(130,10){\circle{4}} %
\put(70,-2){\makebox(0,0)[b]{$n$}} %
\put(60,-15){\makebox(0,0)[b]{$n - 1$}} %
\put(60,-5){\line(0,1){10}} %
\put(80,-5){\line(0,1){10}} %
\put(80,-15){\makebox(0,0)[b]{$n + 1$}} %
\put(3,60){\makebox(0,0)[r]{$t _k$}} %
\put(3,70){\makebox(0,0)[r]{$t _{k + 1}$}} %
\put(3,50){\makebox(0,0)[r]{$t _{k - 1}$}} %
\put(70,60){\circle*{4}} %
\put(70,50){\circle*{4}} %
\put(60,50){\circle*{4}} %
\put(80,50){\circle*{4}} %
\put(60,70){\makebox(0,0){$\times$}} %
\put(70,70){\makebox(0,0){$\times$}} %
\put(80,70){\makebox(0,0){$\times$}} %
\end{picture} 
\end{center}\ \smallskip
\caption{Forward-difference stencil for the approximation to the partial differential equation (\ref {Eq:Model}) at time $t _k$, using the finite-difference scheme (\ref {Eq:FDS}). The black circles represent known approximations to the actual solutions at times $t _{k - 1}$ and $t _k$, and the crosses denote the unknown approximations at time $t _{k + 1}$. \label{Fig:1}}
\end{figure}

\subsection{Nonlinear supratransmission\label {SS:Supratransmission}}

As observed in the literature (see \cite {Geniet-Leon, Geniet-Leon2, Khomeriki}), the double sine-Gordon system (\ref {Eq:Model}), as well as the nonlinear Klein-Gordon and the sine-Gordon chains, and the classical $\beta$-Fermi-Pasta-Ulam systems, presents the phenomenon of supratransmission of energy, which is a nonlinear process characterized by a sudden increase in the amplitude of wave signals propagated into a nonlinear medium by a driving source which irradiates at a frequency in the forbidden band-gap. The mechanism of this transmission of energy is through the generation of localized, nonlinear modes at the driving boundary, in the form of moving breathers or soliton solutions \cite {fabian2009perturbation}.

More concretely, consider a nonlinear system of any of the types mentioned in the previous paragraph, which is perturbed at one end by a harmonic function of the form (\ref {Eq:Driving}), with $\Omega$ a fixed value in the forbidden band-gap of the system. Relatively small driving amplitudes $A$ result in the propagation of practically no energy into the system; however, as the value of $A$ is increased, the existence of a critical value $A _s$, above which the system begins to absorb great amounts of energy from the boundary, is immediately noticed. 

The value $A _s$ introduced in the previous paragraph, is called the {\em supratransmission threshold}, and its existence has been analytically proved for discrete \cite {Geniet-Leon} and continuous \cite{Khomeriki-Leon} sine-Gordon chains, as well as for systems of anharmonic oscillators \cite {Khomeriki}. However, as it was mentioned before, the study of the double sine-Gordon regime has been left practically unexplored.

For our particular study, a simple analysis of the undamped, linearized system of differential equations in (\ref {Eq:Model}) shows that the linear dispersion relation is given by 
\begin{equation}
	\omega ^2 (k) = 1 + 2 c ^2 (1 - \cos k).
	\label{Eq:Dispersion}
\end{equation}
In our simulations, the driving frequency $\Omega$ will take on values in the forbidden band-gap region $\Omega < 1$.

\begin{figure}[tbc]
	\centerline{%
	\begin{tabular}{cc}
		\includegraphics[width=0.48\textwidth]{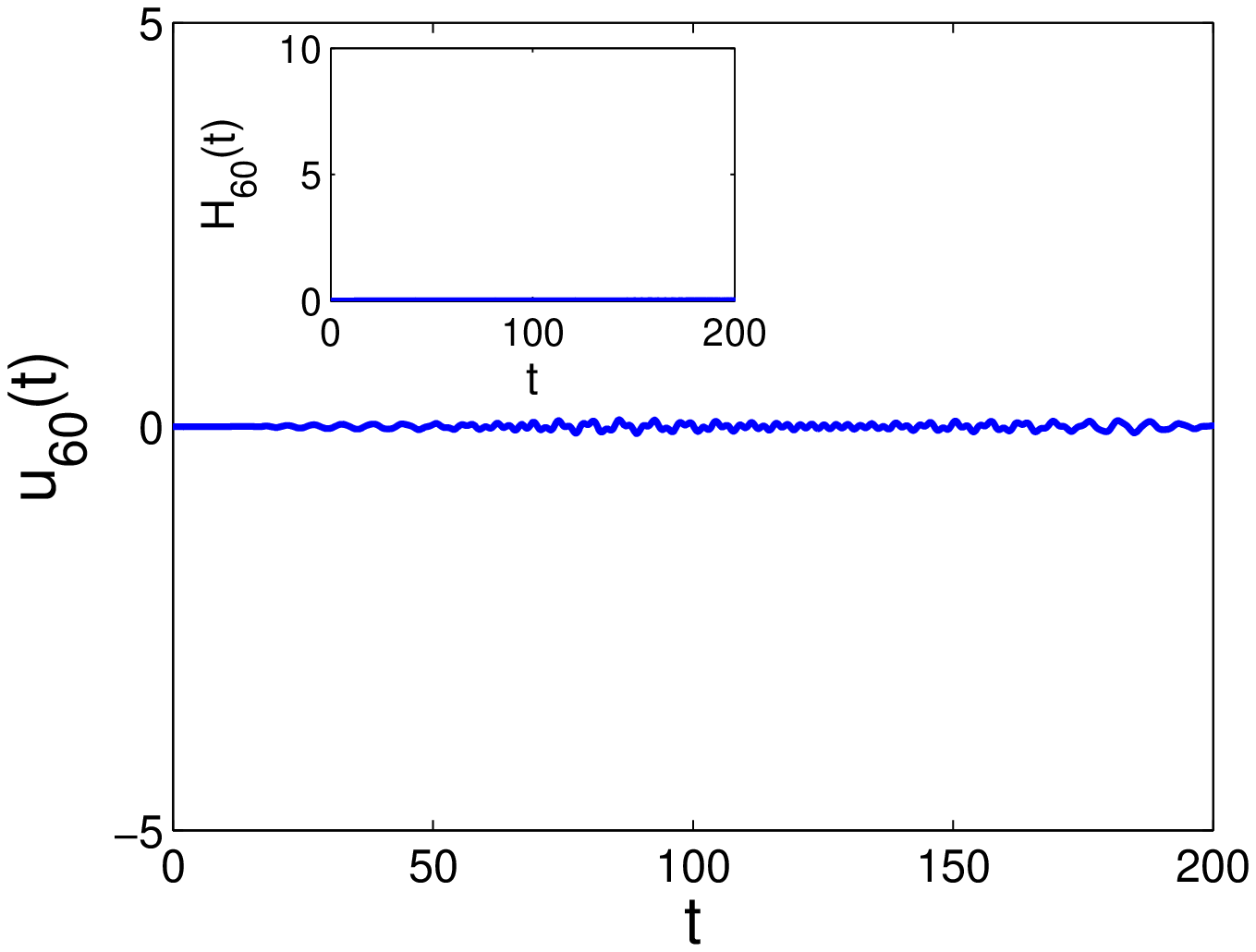} & \includegraphics[width=0.48\textwidth]{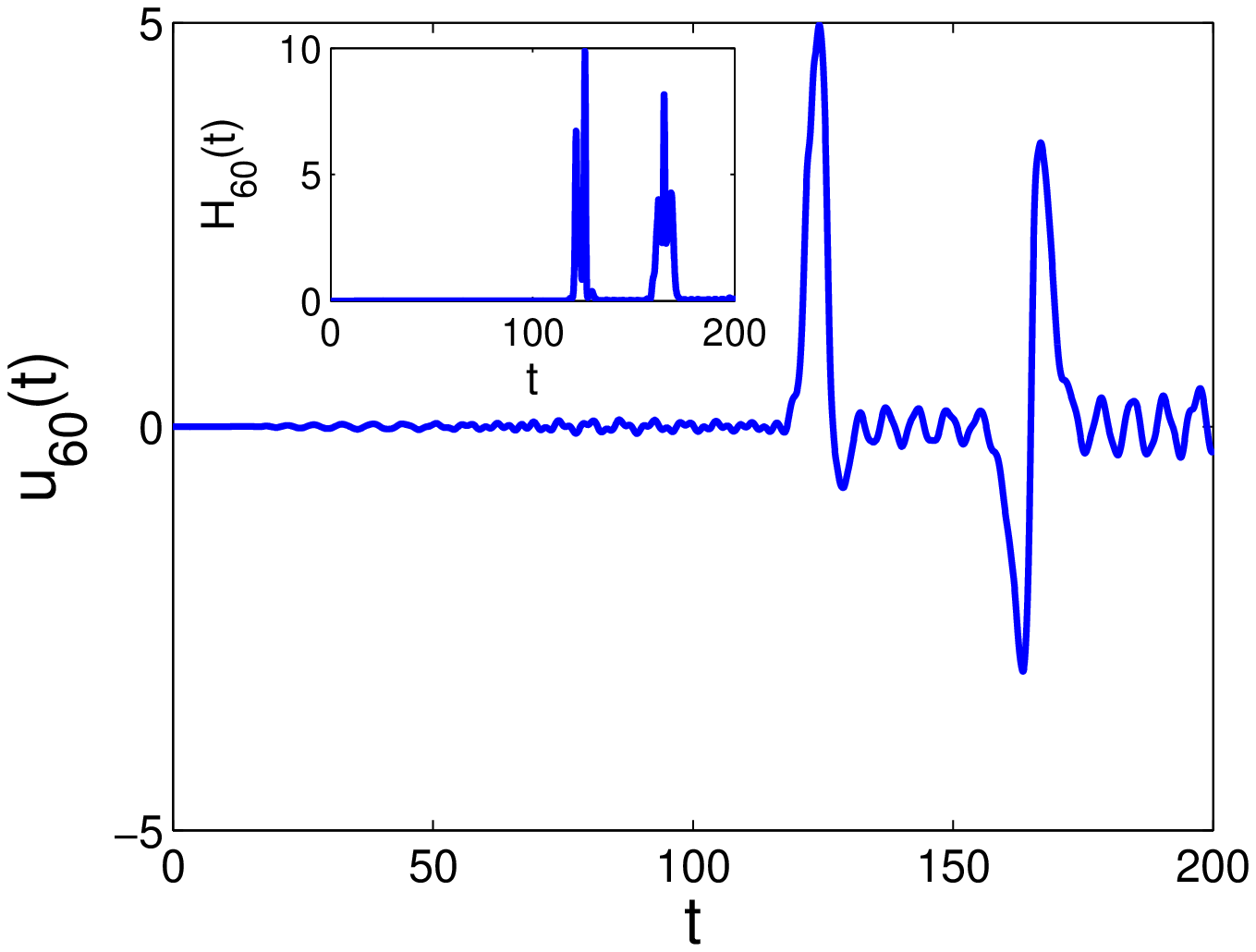}
	\end{tabular}}
	\caption{Graphs of approximate solution $u _{60}$ versus time, of the $60$th node of a system (\ref {Eq:Model}) of length $N = 200$, for $c = 4$, $\gamma = 0$ and a potential $V (u) = 1 - \cos u$. The system was perturbed by means of (\ref {Eq:Driving}) with $\Omega = 0.9$, and two different amplitudes were used: $A = 1.77$ (top graph) and $A = 1.78$ (right column). The insets depict the corresponding temporal evolution of the local energy of the $60$th node. \label {Fig:2}}
\end{figure}

\begin{figure}[tbc]
	\centerline{%
	\begin{tabular}{cc}
		\includegraphics[width=0.48\textwidth]{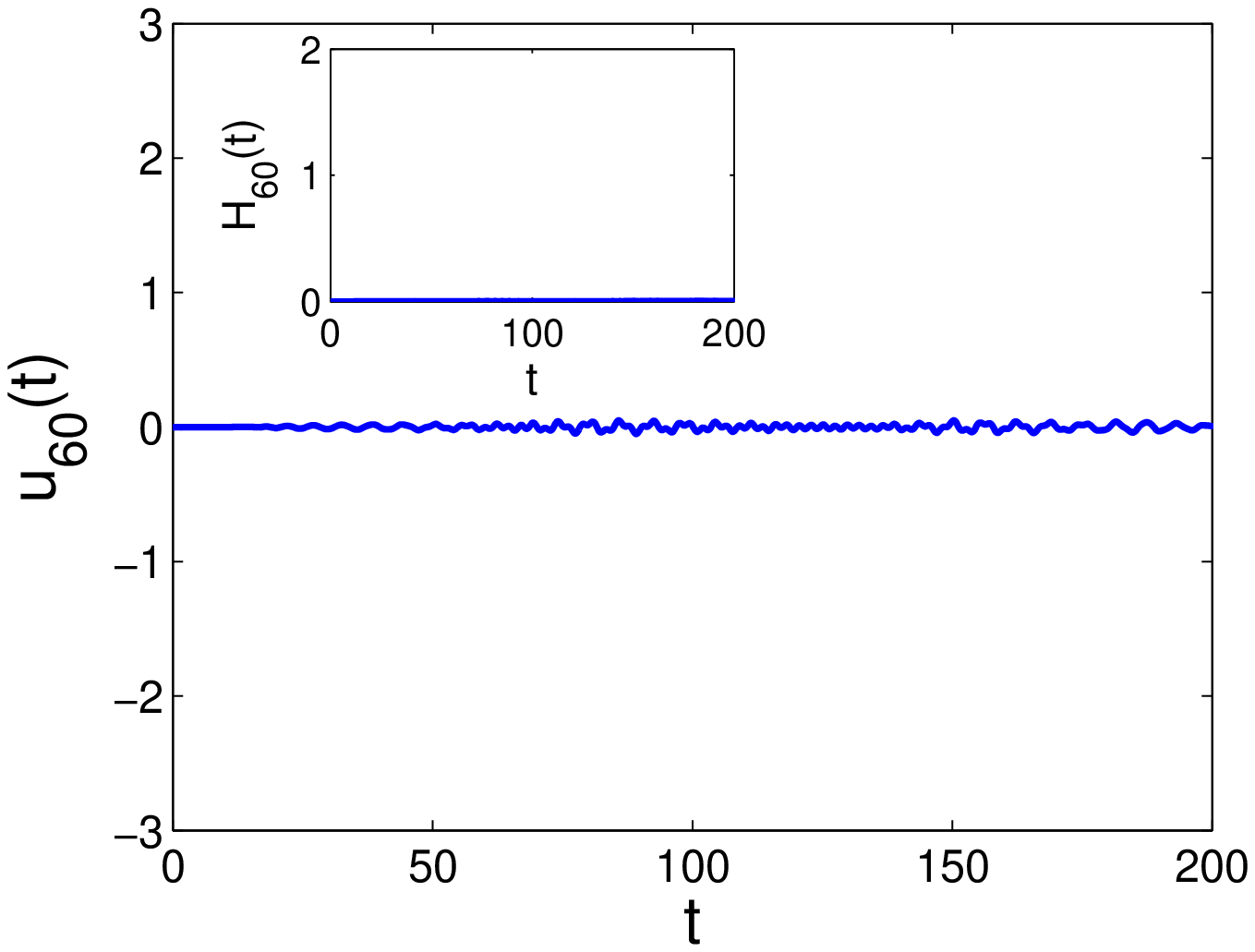} &
		\includegraphics[width=0.48\textwidth]{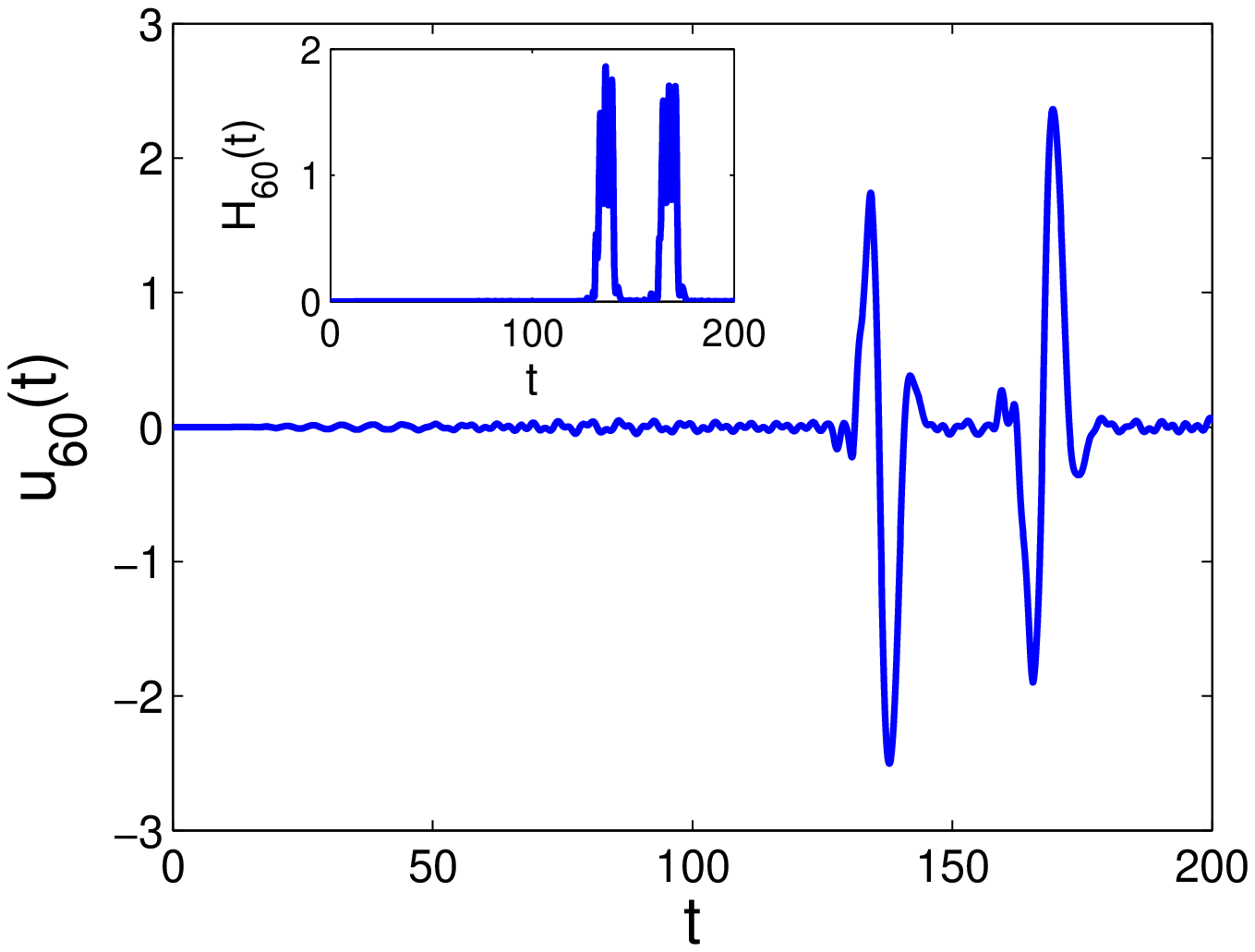}
	\end{tabular}}
	\caption{Graphs of approximate solution $u _{60}$ versus time, of the $60$th node of a system (\ref {Eq:Model}) of length $N = 200$, for $c = 4$, $\gamma = 0$ and a potential given by (\ref {Eq:Potential}). The system was perturbed by means of (\ref {Eq:Driving}) with $\Omega = 0.9$, and two different amplitudes were used: $A = 1.03$ (top graph) and $A = 1.04$ (right column). The insets depict the corresponding temporal evolution of the local energy of the $60$th node. \label {Fig:3}}
\end{figure}

\section{Numerical method\label {S:NumMeth}}

\subsection{Finite-difference scheme\label {SS:FDS}}

Let $N$ be a positive integer, and let $T$ be a positive real number. In order to approximate the solutions of the system (\ref {Eq:Model}) at time $T$, we fix a regular partition of the interval $[0 , T]$ of the form $0 = t _0 < t _1 < \ldots < t _M = T$, of norm $\Delta t = T / M$. Additionally, we let $u _n ^k$ be the numerical approximation of the actual value of $u _n$ at time $t _k$, for $k = 0 , 1 , \dots , M$. Moreover, in order to simplify our notation, we define the temporal differences
\begin{eqnarray}
	\delta _t u _n ^k & = & \frac {u _n ^{k + 1} - u _n ^k} {\Delta t}, \\
	\delta _t ^{(1)} u _n ^k & = & \frac {u _n ^{k + 1} - u _n ^{k - 1}} {2 \Delta t}, \\
	\delta _t ^{(2)} u _n ^k & = & \frac {u _n ^{k + 1} - 2 u _n ^k + u _n ^{k - 1}} {\left( \Delta t \right) ^2},
\end{eqnarray}
for every $n \in \overline {\mathbb {Z}} _N$ and $k \in \mathbb {Z} _M$. Furthermore, for such values of $n$ and $k$, we employ the temporal average operator
\begin{equation}
	\mu _t ^{(1)} u _n ^k = \frac {1} {2} \left( u _n ^{k + 1} + u _n ^{k - 1} \right), 
\end{equation}
and the discrete derivative of $V$ with respect to $u$ and the time average of $V$ at $u$, respectively:
\begin{eqnarray}
	\delta _u ^{(1)} V (u _n ^k) & = & \frac {V (u _n ^{k + 1}) - V (u _n ^{k - 1})} {u _n ^{k + 1} - u _n ^{k - 1}}, \\
	\mu _t V (u _n ^k) & = & \frac {V (u _n ^{k + 1}) + V (u _n ^k)} {2}.
\end{eqnarray}

\begin{figure}[tbc]
	\centerline{%
	\begin{tabular}{cc}
		\includegraphics[width=0.48\textwidth]{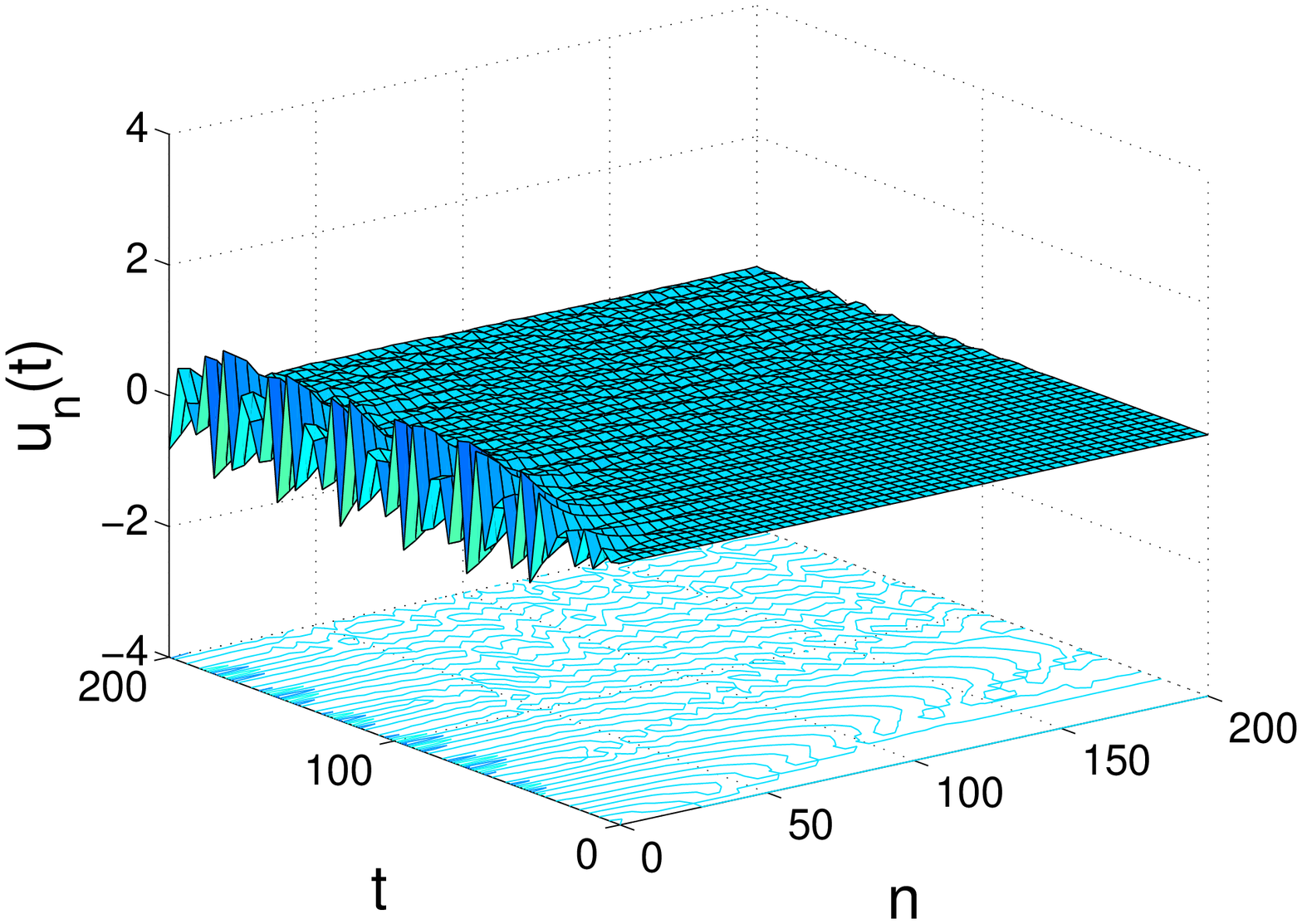} &
		\includegraphics[width=0.48\textwidth]{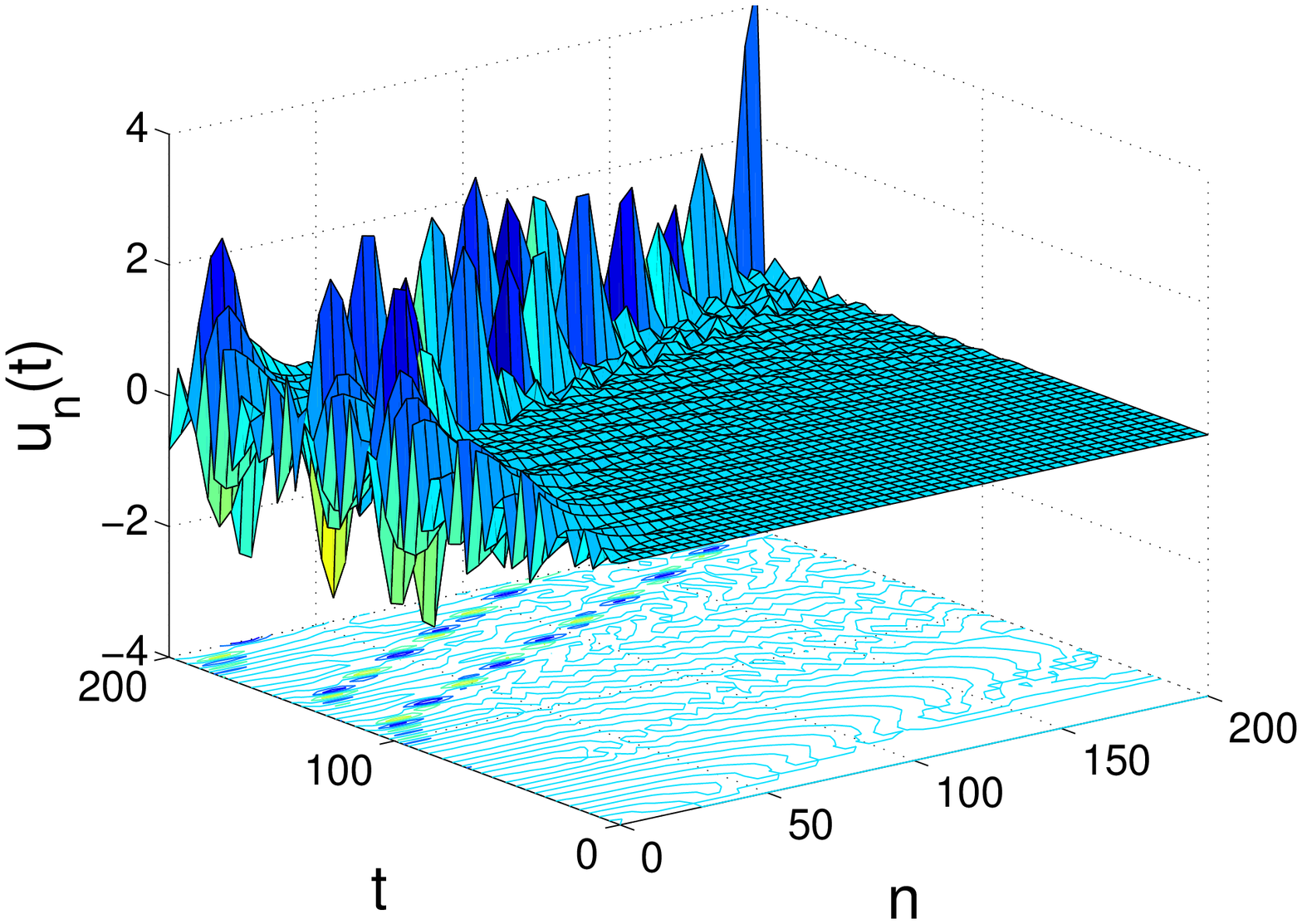}
	\end{tabular}}
	\caption{Graphs of approximate solution versus node site $n$ and time $t$, of a system (\ref {Eq:Model}) of length $N = 200$, for $c = 4$, $\gamma = 0$ and a potential given by (\ref {Eq:Potential}). The system was perturbed by means of (\ref {Eq:Driving}), with $\Omega = 0.9$, and two different amplitude values: $A = 1.03$ (left) and $A = 1.04$ (right). \label {Fig:4}}
\end{figure}

Finally, let 
\begin{equation}
	\phi _k = \phi (t _k).
\end{equation}
With these conventions at hand, the numerical method to approximate solutions of (\ref {Eq:Model}) is summarized as follows:
\begin{equation}
	\begin{array}{c}
			\left( \delta _t ^{(2)} - \mu _t ^{(1)} \delta _x ^{(2)} + \gamma \delta _t ^{(1)} + \delta _u ^{(1)} V \right) (u _n ^k) = 0, \quad \forall n \in \mathbb {Z} _N, \\ 
		\left\{ \begin{array}{ll}
			u _n ^0 = 0, & \forall n \in \mathbb {Z} _N, \\
			u _n ^1 = 0, & \forall n \in \mathbb {Z} _N, \\
			u _0 ^k = \phi _k, & \forall k \in \overline {\mathbb {Z}} _M, \\
			u _N ^k - u _{N - 1} ^k = 0, & \forall k \in \overline {\mathbb {Z}} _M.
		\end{array}\right.
	\end{array}
	\label {Eq:FDS}
\end{equation}
For convenience, the forward-difference stencil of this method has been depicted in Fig. \ref {Fig:1}.

\subsection{Energy scheme\label {SS:EnergyScheme}}

With the same notation as in the previous paragraph, the local energy of the system (\ref {Eq:Model}) at the $n$th node and at the $k$th time step will be approximated by means of the discrete formula
\begin{equation}
		H _n ^k = \frac {1} {2}\left[ \left( \delta _t u _n ^k \right) ^2 + \sum _{j = n - 1} ^n \sum _{l = k} ^{k + 1} \frac {\left( \delta _x u _j ^l \right) ^2} {4} \right] + \mu _t V (u _n ^k),
		\label {Eq:DiscrLocalEnergy}
\end{equation}
where $n \in \mathbb {Z} _N$ and $k \in \mathbb {Z} _M$. Meanwhile, the total energy of the system at time $t _k$ is calculated through the expression
\begin{equation}
	E ^k = \sum _{n \in \mathbb {Z} _N} H _n ^k + \frac {1} {2} \sum _{l = k} ^{k + 1}\frac {\left( \delta _x u _0 ^l \right) ^2} {4}.
	\label {Eq:DiscrTotalEnergy}
\end{equation}

Before closing this stage of our investigation, it is important to point out that the local energy function $H _n$ in (\ref {Eq:LocalEnergy}) is nonnegative for the case of the double sine-Gordon regime; in addition, its discrete counterpart, namely, Eq. (\ref {Eq:DiscrLocalEnergy}), is likewise nonnegative. It follows that the total energy of the system (\ref {Eq:Model}) as given by (\ref {Eq:TotalEnergy}), as well as the discrete total energy (\ref {Eq:DiscrTotalEnergy}) are both nonnegative at any time.

\subsection{Numerical properties\label {SS:NumericalProperties}}

As mentioned previously, the numerical method prescribed by the expressions (\ref {Eq:FDS}), (\ref {Eq:DiscrLocalEnergy}) and (\ref {Eq:DiscrTotalEnergy}) preserves the positivity character of the local and the total energy of the system (\ref {Eq:Model}). Moreover, the finite-difference schemes presented in (\ref {Eq:FDS}) provide consistent solutions of (\ref {Eq:Model}) of order the second order in time (see Appendix \ref {A:Consistency} for a brief discussion of the consistency of the method).

The fact that the local energy estimate (\ref {Eq:DiscrLocalEnergy}) is a consistent approximation of the continuous local energy (\ref {Eq:LocalEnergy}), and that the discrete total energy (\ref {Eq:DiscrTotalEnergy}), in turn, is a consistent estimation of the corresponding continuous expression (\ref {Eq:TotalEnergy}), is evident. The following result shows that this consistency is also preserved on the grounds of the rate of change of energy with respect to time.

\begin{proposition}
	Consider the finite-difference scheme (\ref {Eq:FDS}), with local energy given by (\ref {Eq:DiscrLocalEnergy}), and total energy (\ref{Eq:DiscrTotalEnergy}). Then, the discrete rate of change of energy of the method at time $t _{k - 1}$ is given by
	\begin{equation}
		\delta _t E ^{k - 1} = - c \left( \mu _t ^{(1)} \delta _x u _0 ^k \right) \left( \delta _t ^{(1)} u _0 ^k \right) - \gamma \sum _{n \in \mathbb {Z} _N} \left( \delta _t ^{(1)} u _n ^k \right) ^2 
	\end{equation} 
	\label{Prop:DerTotEner}
\end{proposition}

\begin{proof}
See Appendix \ref {A:Proof}.
\end{proof}

\subsection{Computational remarks\label {SS:ComputRemarks}}

Clearly, the finite-difference scheme (\ref {Eq:FDS}) is nonlinear and implicit when $V$ is not a constant function, as it is the case of the double sine-Gordon chain. Thus, in order to approximate the solution of the system (\ref {Eq:Model}) at time $t _{k + 1}$ when the approximations at times $t _k$ and $t _{k - 1}$ are at hand, we employ Newton's method for nonlinear systems of equations.

\begin{figure}[tbc]
	\centerline{%
	\begin{tabular}{cc}
		\includegraphics[width=0.48\textwidth]{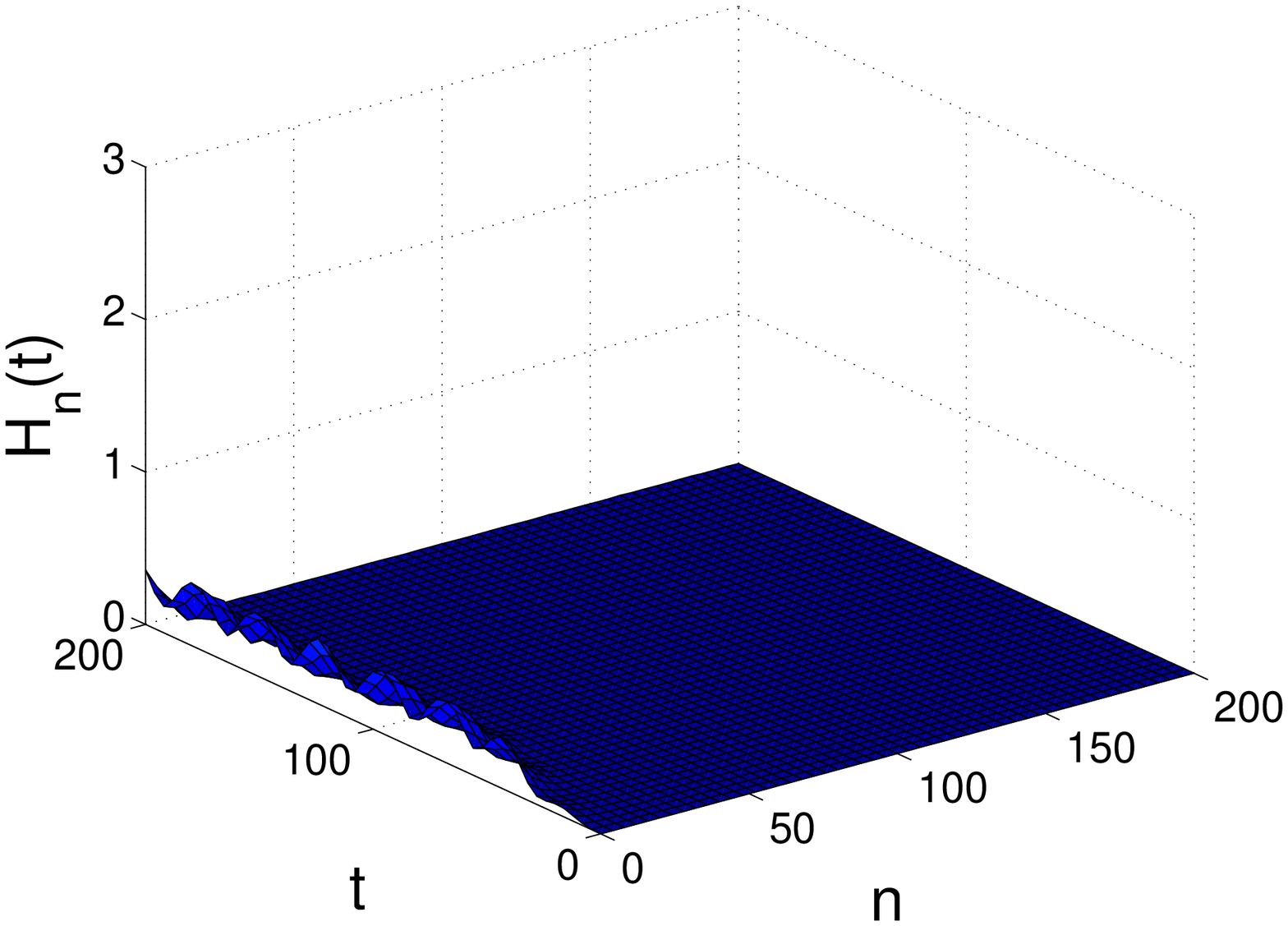} &
		\includegraphics[width=0.48\textwidth]{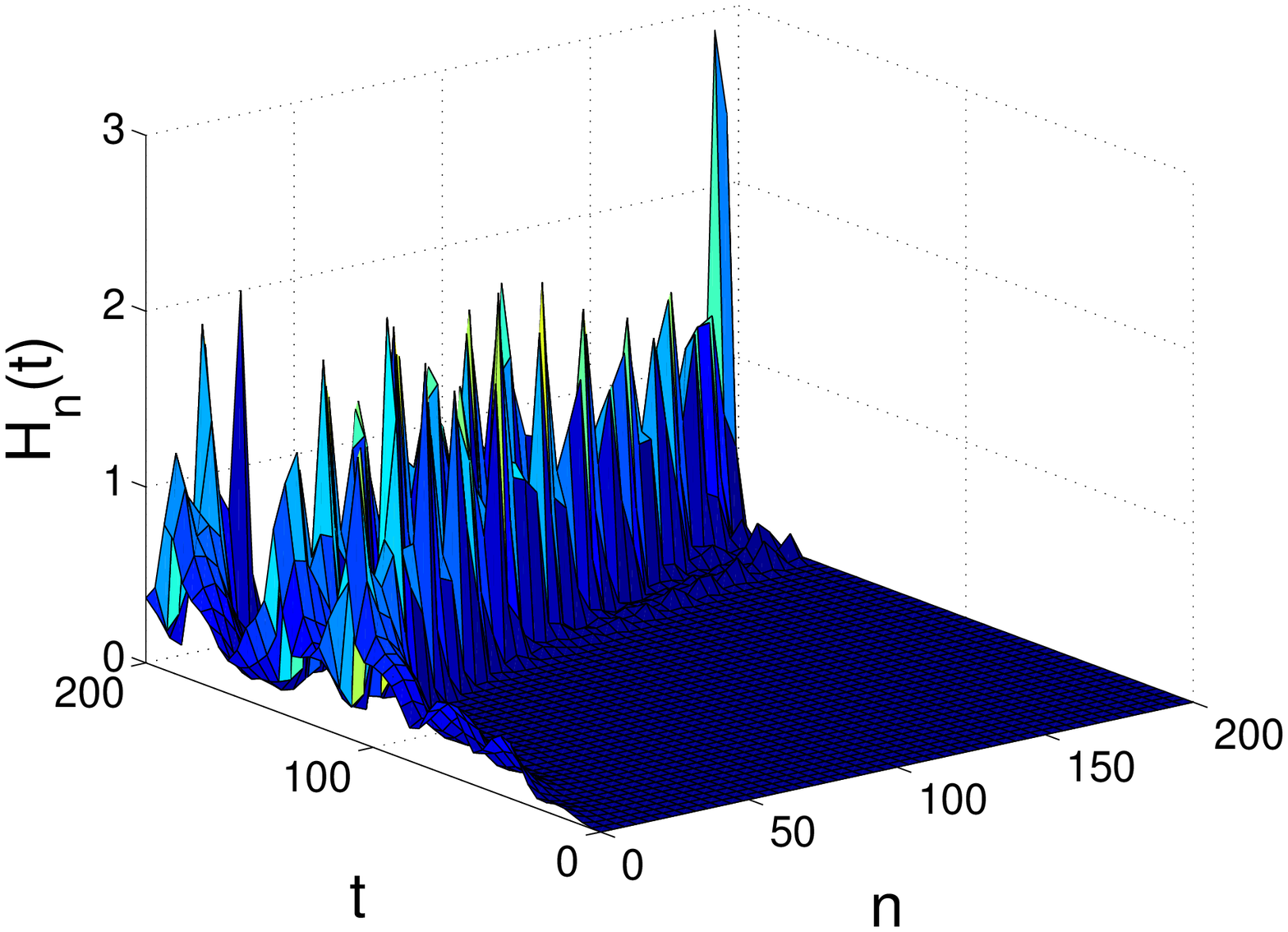}
	\end{tabular}}
	\caption{Graphs of approximate local energy versus node site $n$ and time $t$, of a system (\ref {Eq:Model}) of length $N = 200$, for $c = 4$, $\gamma = 0$ and a potential given by (\ref {Eq:Potential}). The system was perturbed by means of (\ref {Eq:Driving}), with $\Omega = 0.9$, and two different amplitude values: $A = 1.03$ (top graph) and $A = 1.04$ (right column). \label {Fig:5}}
\end{figure}

\begin{figure}
	\centerline{\includegraphics[width=0.8\textwidth]{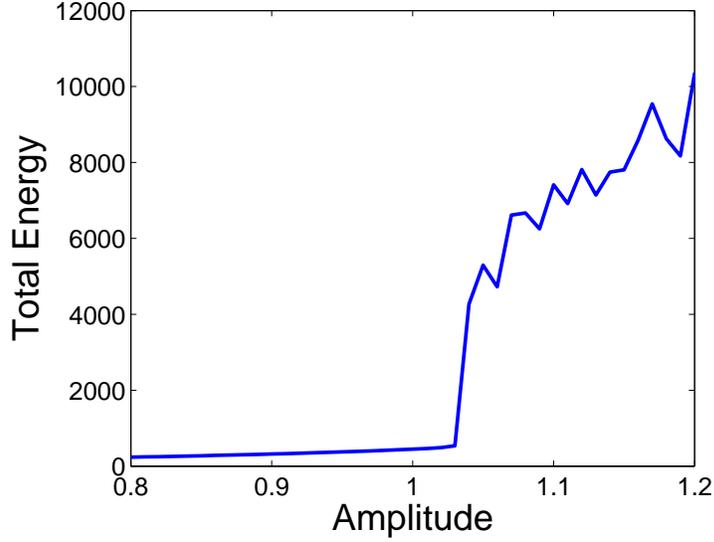}}
	\caption{Graph of approximate total energy over the time period $[0 , 200]$ of the undamped system (\ref {Eq:Model}) versus driving amplitude, subject to harmonic driving of the form (\ref {Eq:Driving}) and a potential (\ref {Eq:Potential}). The parameters $c = 4$, $N = 200$ and $\Omega = 0.9$ were employed in the simulations. \label {Fig:6}}
\end{figure}

Once, again, let us adopt the conventions of Section \ref {SS:FDS}. For every $k \in \mathbb {Z} _M$, let $\mathbf {u} ^k = (u _0 ^k , u _1 ^k , \ldots , u _N ^k)$,
and let $f _n$ be the left-hand side of the $n$th difference equation of (\ref {Eq:FDS}), that is, let
\begin{equation}
	f _n (\mathbf {u} ^k) = \left[ \delta _t ^{(2)} - \mu _t ^{(1)} \delta _x ^{(2)} + \gamma \delta _t ^{(1)} + \delta _u ^{(1)} V \right] (u _n ^k).
\end{equation}
for every $n \in \mathbf {Z} _N$. Additionally, let
\begin{eqnarray}
	f _0 (\mathbf {u} ^k) & = & u _0 ^k - \phi _k, \\
	f _N (\mathbf {u} ^k) & = & u _N ^k - u _{N - 1} ^k.
\end{eqnarray}
Moreover, let $\mathbf {f} = (f _0 , f _1 , \ldots , f _N)$. Using a recursive process, assume that the vectors $\mathbf {u} ^k$ and $\mathbf {u} ^{k - 1}$ have been previously computed. Then 
\begin{equation}
	\mathbf {u} ^{k + 1} = \mathbf {u} ^k - \mathbf {y}, 
	\label{Eq:Iterative}
\end{equation}
where $\mathbf {y}$ is the $(N + 1)$-dimensional vector which satisfies the matrix equation
\begin{equation}
	J (\mathbf {u} ^k) \mathbf {y} = - \mathbf {f} (\mathbf {u} ^k).
	\label {Eq:MatrixSystem}
\end{equation}

Evidently, the $(N + 1) \times (N + 1)$ matrix $J$ is the Jacobian matrix of $\mathbf {f}$, which is given by
\begin{equation}
	J (\mathbf {u} ^k) = \left( \begin{array}{cccccccc}
		1 & 0 & 0 & 0 & \cdots & 0 & 0 & 0 \\
		a & d _1 & a & 0 & \cdots & 0 & 0 & 0 \\
		0 & a & d _2 & a & \cdots & 0 & 0 & 0 \\
		\vdots & \vdots & \vdots & \vdots & \ddots & \vdots & \vdots & \vdots \\
		0 & 0 & 0 & 0 & \cdots & a & d _{N - 1} & a \\
		0 & 0 & 0 & 0 & \cdots & 0 & - 1 & 1
		\end{array} \right),
\end{equation}
where
\begin{eqnarray}
	a & = & - \frac {c ^2} {2}, \\
	d _n & = & \frac {1} {\left( \Delta t \right) ^2} + c ^2 + \frac {\gamma} {2 \Delta t} \\
	 & & + \frac {\left( u _n ^{k + 1} - u _n ^{k - 1} \right) V ^\prime (u _n ^{k + 1}) + V (u _n ^{k - 1}) - V (u _n ^{k + 1})} {\left(u _n ^{k + 1} - u _n ^{k - 1}\right) ^2}, \nonumber
\end{eqnarray}
for every $n \in \mathbb {Z} _N$. The tridiagonal system (\ref {Eq:MatrixSystem}) is solved then employing Crout's reduction technique with pivoting \cite{Burden}.

Of course, for our simulations, Newton's method requires of a stopping criterion in order to approximate the vector $\mathbf {u} ^{k + 1}$ in Eq. (\ref {Eq:Iterative}). Particularly, in this work, this criterion is given by the condition $\Vert \mathbf {y} \Vert _2 < \epsilon$, where the tolerance parameter $\epsilon$ is equal to $1 \times 10 ^{- 4}$, and $\Vert \cdot \Vert _2$ is the classical Euclidean norm in $\mathbb {R} ^{N + 1}$.

\section{Simulations\label {S:Simul}}

Throughout this section, we consider a system governed by (\ref {Eq:Model}), where the driving function assumes the sinusoidal form (\ref {Eq:Driving}). In order to avoid the creation of shock waves produced by the sudden movement of the driving boundary at the initial time, we linearly increase the driving amplitude from zero to its actual value $A$ during a finite period of time $T _0$. Particularly, in the simulations performed in this study, we fix $T _0 = 50$.

\subsection{Sine-Gordon chain\label {SS:Sine-Gordon}}

\begin{figure}[tbc]
	\centerline{%
	\begin{tabular}{cc}
	\includegraphics[width=0.48\textwidth]{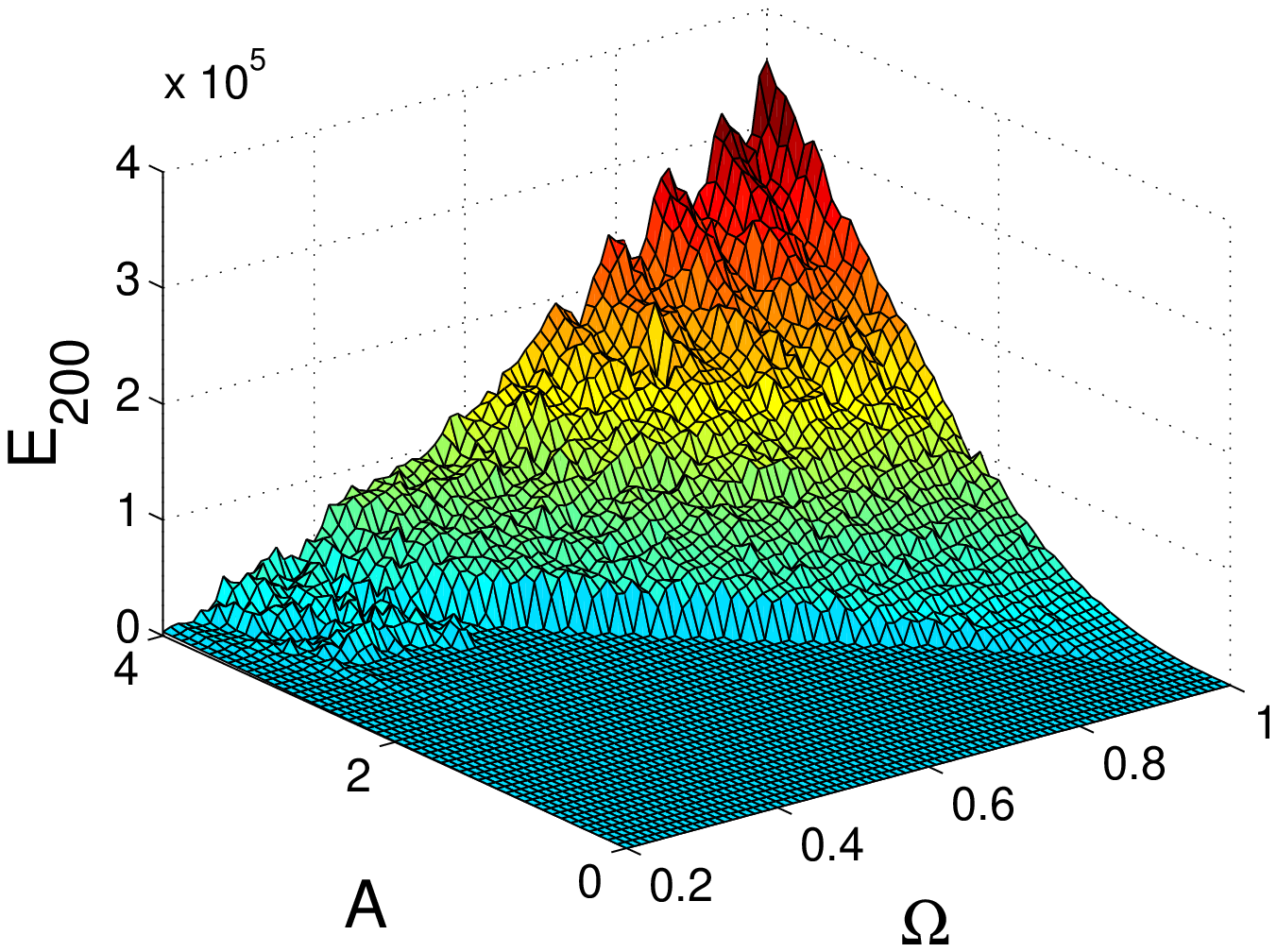} &
	\includegraphics[width=0.48\textwidth]{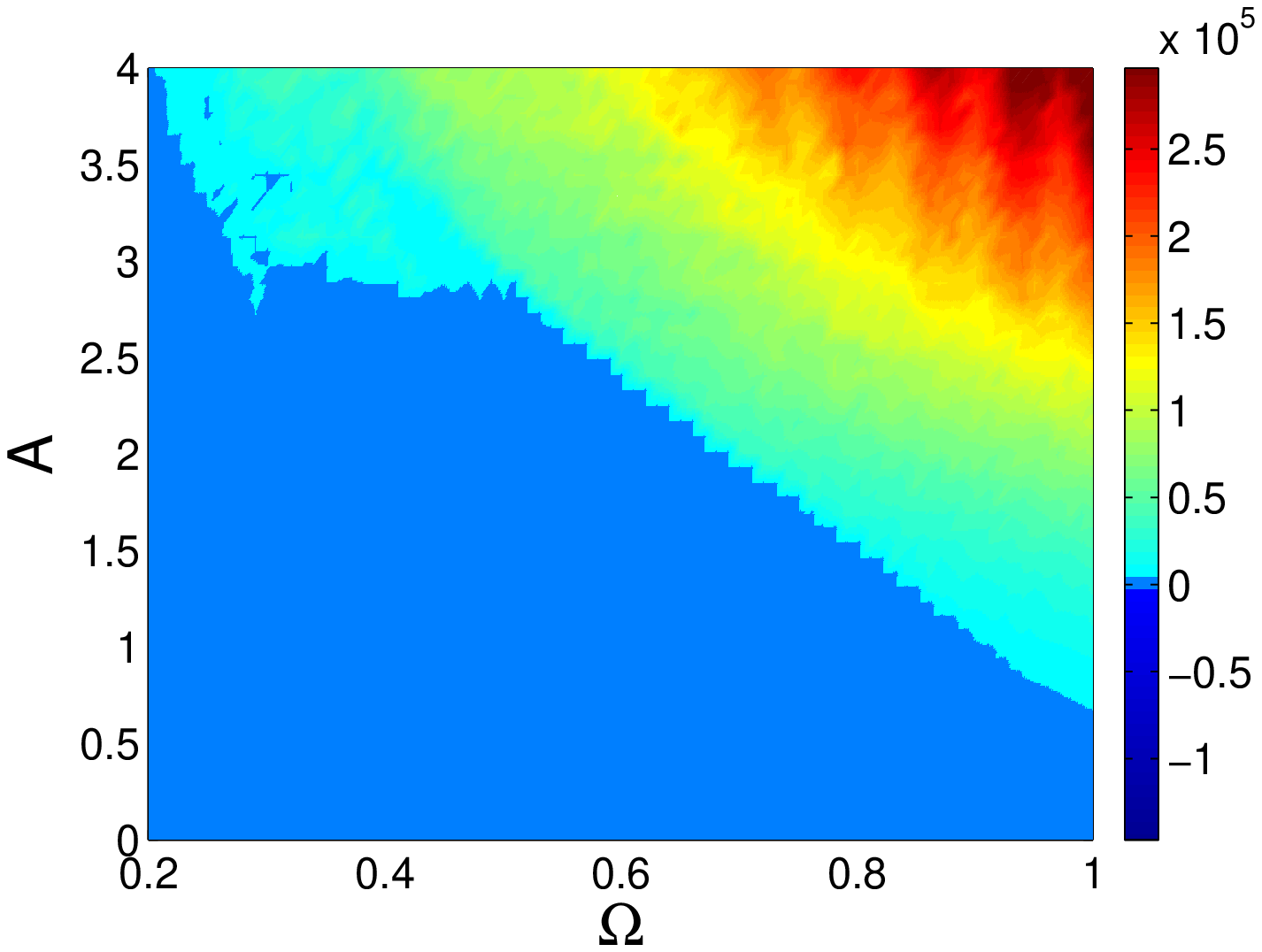}
	\end{tabular}}
	\caption{Graph of approximate total energy over the time period $[0 , 200]$ of the undamped system (\ref {Eq:Model}) versus driving amplitude and driving frequency (left), subject to harmonic driving of the form (\ref {Eq:Driving}) and a potential (\ref {Eq:Potential}). The parameters $c = 4$ and $N = 200$ were employed in the simulations. The left graph is the checkboard plot of the top one. \label {Fig:7}}
\end{figure}

As a means to verify the validity of our method, we consider, first of all, a discrete chain of harmonic oscillators coupled through identical springs, in which case, the governing equations are given by (\ref {Eq:Model}), with $V (u) = 1 - \cos u$. Moreover, assume that the system under study is undamped, let $c = 4$, $N = 200$ and fix a driving frequency of $0.9$. According to \cite {Geniet-Leon}, the supratransmission threshold of the system occurs around the critical value $A _s = 1.78$.

\begin{figure}
	\centerline{\includegraphics[width=0.8\textwidth]{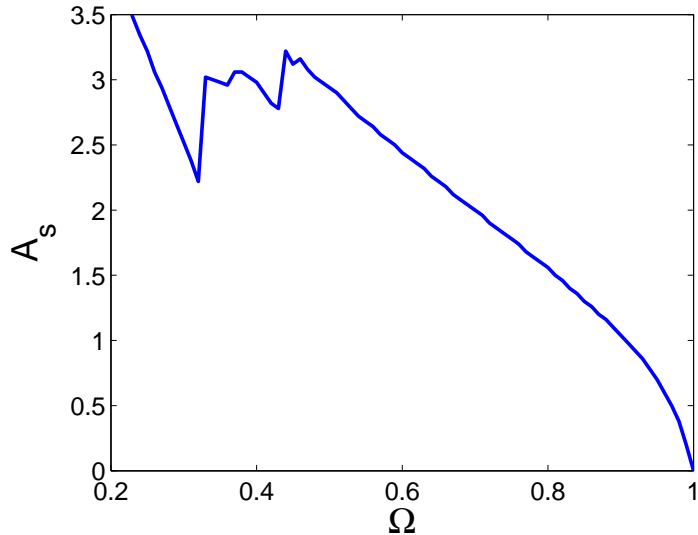}}
	\caption{Graph of approximate driving amplitude $A _s$ above which supratransmission occurs in the undamped system (\ref {Eq:Model}) versus driving frequency $\Omega$, subject to harmonic driving of the form (\ref {Eq:Driving}) and a potential (\ref {Eq:Potential}). The parameters $c = 4$ and $N = 200$ were employed in the simulations. \label {Fig:8}}
\end{figure}

From a computational perspective, we let $\Delta t = 0.05$, and compute approximations to the actual solution of the initial-boundary-value probme (\ref {Eq:Model}) and the corresponding local energy, over a time interval of length $T = 200$. Under these circumstances, Fig. \ref {Fig:2} presents the temporal evolution of the solution and the local energy of the $60$th node of the system for two values of the driving amplitude: $A = 1.77$ (top graph) and $A = 1.78$ (below graph). The results show a drastic change in the qualitative behavior of solutions around the proposed critical amplitude $A _s$. These results are clearly in agreement with \cite {Geniet-Leon}, and they are considered as evidence in favor of both the validity of our method and the existence of supratransmission in the sine-Gordon chain.

\subsection{Double sine-Gordon chain\label {SS:DSine-Gordon}}

As mentioned before, the study of the phenomenon of nonlinear supratransmission of energy in the double sine-Gordon chain is a topic of interest that has been left aside. In this section, however, we proceed to compute bifurcation diagrams similar those constructed to predict the process of supratransmission in discrete sine-Gordon and Klein-Gordon systems (see \cite{Geniet-Leon, Geniet-Leon2}). So, as in the previous stage of our investigation, we consider an undamped system governed by (\ref {Eq:Model}), with parameters $c = 4$, $N = 200$, $T = 200$, $\Omega = 0.9$, and potential given by (\ref {Eq:Potential}). Computationally, let $\Delta t = 0.05$.

With these considerations, Fig. \ref {Fig:3} presents the time-dependent graphs of the solution and the local energy of the $60$th node of the system, for two different values of the driving amplitude, namely, $A = 1.03$ (top graph) and $A = 1.04$ (bottom graph). As in the case of the discrete sine-Gordon system, we observe a drastic qualitative difference in the behavior of the solution and the local energy of the $60$th node, around the critical value $A _s = 1.04$. Indeed, this observation is in agreement with the available literature \cite {Geniet-Leon2}.

\begin{figure}[tbc]
	\centerline{\includegraphics[width=0.8\textwidth]{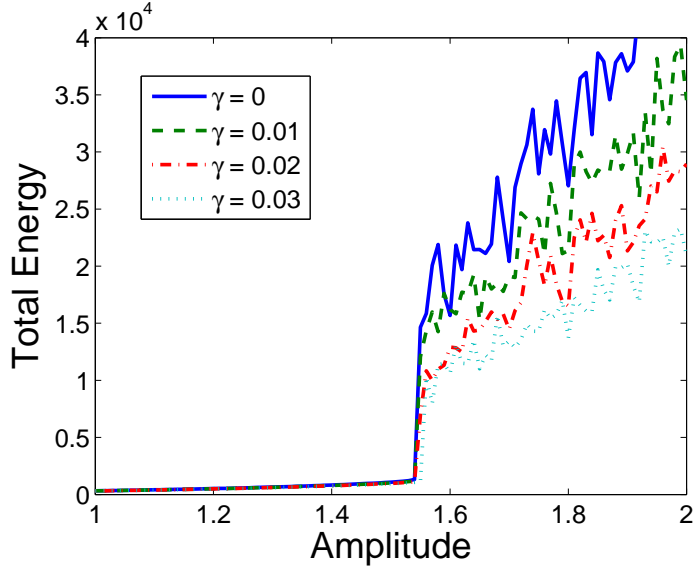}}
	\caption{Graph of approximate total energy over the time period $[0 , 200]$ of the system (\ref {Eq:Model}) versus driving amplitude, subject to harmonic driving of the form (\ref {Eq:Driving}) and a potential (\ref {Eq:Potential}). The parameters $c = 4$, $N = 200$ and $\Omega = 0.9$ were employed in the simulations, with four different values of the damping coefficient: $\gamma = 0$ (solid), $\gamma = 0.01$ (dashed), $\gamma = 0.02$ (dash-dotted), $\gamma = 0.03$ (dotted). \label {Fig:9}}
\end{figure}

Fig. \ref {Fig:4} presents the solution of the system studied in the previous paragraph, with respect to node site $n$ and time $t$, where $t \in [0 , 200]$. The graphs of the solutions clearly change drastically for the two driving amplitudes considered: for $A = 1.03$, the boundary obviously does not propagate wave signals into the system; on the contrary, the graph corresponding to $A = 1.04$ shows transmission of energy into the medium. This observation is verified in Fig. \ref {Fig:5}, which presents the corresponding graphs of local energy for the two amplitudes considered. Evidently, the qualitative observations done in the domain of the solutions carry over to the domain of the local energy of the system.

For the next step in our discussion, we define the total energy of the system (\ref {Eq:Model}) over the time interval $[0 , T]$ as
\begin{equation}
	E _T = \int _0 ^T E (t) d t,
\end{equation}
where $E (t)$ is the total energy of the system at time $t$, given by expression (\ref {Eq:TotalEnergy}). Clearly, $E _T$ is consistently approximated by means of the formula
\begin{equation}
	E _T ^\prime = \sum _{k = 1} ^{M - 1} E ^k \Delta t,
\end{equation}
where each $E ^k$ is given by (\ref {Eq:DiscrTotalEnergy}).

With this notation, Fig. \ref {Fig:6} presents the total energy over the temporal interval $[0 , 200]$, of a system (\ref {Eq:Model}) with the same parameters as above, when the driving amplitude takes on values in the interval $[0.8 , 1.3]$. The graph evidently shows the drastic change in the behavior of the total energy of the system before and after the amplitude value $1.04$. With this strong evidence of the existence of supratransmission in the double sine-Gordon chain, the critical value $A _s = 1.04$ is immediately identified as the nonlinear supratransmission threshold for $\Omega = 0.9$. Obviously, these results are in perfect agreement with \cite {Geniet-Leon2}.

We have performed similar simulations for several values of $\Omega$ in the interval $[0.2 , 1]$, and values of $A$ in $[0 , 4]$, and we have obtained qualitatively identical results. Indeed, Fig. \ref {Fig:7} summarizes our findings, in the form of a graph of total energy over the period of time $[0 , 200]$, versus driving frequency and driving amplitude. Clearly, for every such driving frequency, there exists a smallest driving amplitude $A _s$ above which the system begins to absorb energy from the boundary. From here, a graph of $A _s$ versus driving frequency is obtained and presented as Fig. \ref {Fig:8}. It is worth noticing that the results are in qualitative agreement with those obtained for sine-Gordon chains \cite {Geniet-Leon, Geniet-Leon2}.

Finally, it must be mentioned that the numerical method employed in this work is also useful in order to establish the effects of damping in a discrete double sine-Gordon chain governed by (\ref {Eq:Model}). Indeed, consider a system consisting of $N = 200$ nodes coupled through (\ref {Eq:Model}), with $c = 4$, harmonically perturbed by the driving function (\ref {Eq:Driving}) with $\Omega = 0.8$, over an interval of time $[0 , 200]$, where the potential function $V$ is given by (\ref {Eq:Potential}). Fig. \ref {Fig:9} presents the effect of the driving amplitude $A$ on the total energy of the system, for values of $A$ in the interval $[1 , 2]$, and three different values of the damping coefficient, namely, $\gamma = 0$, $\gamma = 0.01$, $\gamma = 0.02$ and $\gamma = 0.03$. The results show the expected decrease in the total energy of the system as $\gamma$ is  increased and, moreover, they show that the supratransmission threshold is slightly delayed with the presence of damping.

\section{Conclusions\label {S:Concl}}

In this work, we have employed a numerical method in the study of the occurrence of the process of nonlinear supratransmission in a discrete chain of oscillators coupled with identical springs. The method proposed is consistent of order $\mathcal {O} ((\Delta t) ^2)$, and it is associated to a discrete scheme to approximate the local energy of the chain, as well as a scheme for the total energy of the system. Both energy schemes consistently approximate their continuous counterparts, and the method has the property that the discrete rate of change of energy also approximates the corresponding continuous rate of change.

The method was qualitatively tested against known approximations to the occurrence of the phenomenon of nonlinear supratransmission in discrete sine-Gordon and double sine-Gordon chains. The simulations obtained with our method are indeed in excellent agreement with the results available in the literature. Moreover, the method was employed in the construction of a bifurcation diagram of smallest driving amplitude at which supratransmission starts in the undamped system, versus driving frequency. The graph is actually in qualitative agreement with those found in the literature for discrete sine-Gordon and Klein-Gordon chains, which are systems with the same forbidden band-gap region. Moreover, when damping is present, our simulations show that the process of supratransmission is still present in the system under investigation, and that the appearance of the critical amplitude value is delayed as the damping coefficient increases. 

Of course, many avenues of research still remain open. Thus, from a practical point of view, it is important to provide applications of the results presented in this work. More concretely, following \cite {Khomeriki-Leon}, it is interesting to propose applications of the process of nonlinear supratransmission to the design of amplifiers of weak signals, or to the fabrication of detectors of ultra weak pulses, as it has been done for the Klein-Gordon equation \cite {Chevrieux2, Khomeriki-Leon}.

\subsubsection*{Acknowledgments}

The author would like to acknowledge the enlightening comments of the anonymous reviewers, which led to improve the overall quality of the final version of this manuscript. Also, he would like to thank Dr. F. J. \'{A}lvarez Rodr\'{\i}guez, dean of the Faculty of Sciences at the Universidad Aut\'{o}noma de Aguascalientes, and Dr. F. J. Avelar Gonz\'{a}lez, Director of the Office for Graduate Studies and Research of the same university, for uninterestedly providing the computational resources to produce this article. This work presents the final results of project PIM08-1 at this university. 

\appendix

\section{Consistency study\label{A:Consistency}}

A brief consistency analysis of the finite-difference schemes (\ref {Eq:FDS}) reveals that the numerical method proposed in this work is consistent of order $\mathcal {O} ((\Delta t) ^2)$. In fact, observe that
\begin{eqnarray}
	\delta _t ^{(1)} u _n ^k & \approx & \frac {d u _n} {d t} (t _k) + \frac {(\Delta t) ^2} {12} \frac {d ^3 u _n} {d t ^3} (t _k), \\
	\delta _t ^{(2)} u _n ^k & \approx & \frac {d ^2 u _n} {d t ^2} (t _k) + \frac {(\Delta t) ^2} {12} \frac {d ^4 u _n} {d t ^4} (t _k),
\end{eqnarray}
for every $n \in \mathbb {Z} _N$ and every $k \in \mathbb {Z} _M$. Moreover,
\begin{equation}
	\mu _t ^{(1)} \delta _x ^{(2)} u _n ^k \approx \delta _x ^{(2)} u _n ^k + \frac {(\Delta t) ^2} {2} \delta _x ^{(2)} \frac {d ^2 u _n} {d t ^2}.
\end{equation}

\section{Energy consistency \label{A:Proof}}

For the sake of simplification, we introduce the following notation, for every $n \in \mathbb {Z} _N$ and $k \in \mathbb {Z} _M$:
\begin{eqnarray}
	\mu _x u _n ^k & = & \displaystyle {\frac {1} {2} \left( u _{n + 1} ^k + u _n ^k \right)}, \\
	\iota ^k & = & \displaystyle {\frac {1} {4} \sum _{l = k} ^{k + 1}\frac {\left( \delta _x u _0 ^l \right) ^2} {2}},\\
	h _n ^k & = & \displaystyle {\frac {1} {2} \sum _{j = n - 1} ^n \sum _{l = k} ^{k + 1} \frac {\left( \delta _x u _j ^l \right) ^2} {4}}.
\end{eqnarray}
Clearly, the term $\iota ^k$ is identified with the {\em independent term} (the term which is not prescribed by the summation over all $n \in \mathbb {Z} _N$) to the right-hand side of Eq. (\ref {Eq:DiscrTotalEnergy}).

\begin{proof}[Proof of Proposition \ref {Prop:DerTotEner}]
It is easy to check that the following identities are valid for every $n \in \mathbb {Z} _N$ and $k \in \mathbb {Z} _M$:
\begin{equation}
	\frac {1} {2} \left(\delta _t u _n ^k\right) ^2 - \frac {1} {2} \left(\delta _t u _n ^{k - 1}\right) ^2 = \left(\delta _t ^{(2)} u _n ^k\right) \left(\delta _t ^{(1)} u _n ^k \right) \Delta t,
\end{equation}
\begin{equation}
	\mu _t V (u _n ^k) - \mu _t V (u _n ^{k - 1}) = \left( \delta _u ^{(1)} V (u _n ^k) \right) \left( \delta _t ^{(1)} u _n ^k \right) \Delta t.
\end{equation}
It is a tedious algebraic task (though relatively easy) to verify that the following equalities hold, for every $n \in \mathbb {Z} _N$ and $k \in \mathbb {Z} _M$:
\begin{equation}
	\begin{array}{rcl}
		\delta _t h _n ^{k - 1} & = & - \frac {1} {2} \left( \delta _t ^{(1)} u _n ^k \right) \left( \mu _t ^{(1)} \delta _x ^{(2)} u _n ^k \right) \\
		 & & + \frac {c} {2} \left( \delta _t ^{(1)} u _{n + 1} ^k \right) \left( \mu _t ^{(1)} \delta _x u _n ^k \right) \\
	 	 & & - \frac {c} {2} \left( \delta _t ^{(1)} u _{n - 1} ^k \right) \left( \mu _t ^{(1)} \delta _x u _{n - 1} ^k \right).
	 \end{array}
\end{equation}
As a consequence, 
\begin{equation}
	\begin{array}{rcl}
		\displaystyle {\sum _{n \in \mathbb {Z} _N} \delta _t H _n ^{k - 1}} & = & \displaystyle {\sum _{n \in \mathbb {Z} _N} \left\{ \left[ \delta _t ^{(2)} - \mu _t ^{(1)} \delta _x ^{(2)} + \delta _u ^{(1)} V \right] (u _n ^k) \cdot \right.} \\
		 & & \displaystyle { \left. \left( \delta _t ^{(1)} u _n ^k \right) \right\} - c \left( \mu _x \delta _t ^{(1)} u _0 ^k \right) \left( \mu _t ^{(1)} \delta _x u _0 ^k \right)} \\
		 & = & \displaystyle {- \gamma \sum _{n \in \mathbb {Z} _N} \left( \delta _t ^{(1)} u _n ^k \right) ^2} \\
		 & & \displaystyle { - c \left( \mu _x \delta _t ^{(1)} u _0 ^k \right) \left( \mu _t ^{(1)} \delta _x u _0 ^k \right)}.
	\end{array}
\end{equation}
Moreover, 
\begin{equation}
	\delta _t \iota ^{k - 1} = \frac {1} {2} \left( \delta _t ^{(1)} \delta _x u _0 ^k \right) \left( \mu _t ^{(1)} \delta _x u _0 ^k \right).
\end{equation}
The conclusion of Proposition \ref {Prop:DerTotEner} is now reached by computing $\delta _t E ^{k - 1}$ and simplifying.
\end{proof}

\end{document}